\documentclass{ecai}
\usepackage{times}
\usepackage{graphicx}
\usepackage{latexsym}

\usepackage{soul}
\usepackage{url}
\usepackage[hidelinks]{hyperref}
\usepackage[utf8]{inputenc}
\usepackage[small]{caption}
\usepackage{amssymb,amsmath,amsthm}
\usepackage{booktabs}
\usepackage{named}
\usepackage{algorithm}
\usepackage{algorithmic}
\usepackage{color}

\usepackage{authblk}
\usepackage{blindtext}

\newtheorem{theorem}{Theorem}
\newtheorem{definition}{Definition}

\newtheorem{corollary}{Corollary}
\newtheorem{lemma}{Lemma}

\begin{document}

\title{Optimizing Off-Chain Payment Networks in Cryptocurrencies}

\author{Yotam Sali and Aviv Zohar}
\affil{The Hebrew University of Jerusalem, The School of Computer Science}
\affil{\textit {\{yotam.sali,avivz\}@cs.huji.ac.il}}

\maketitle
\bibliographystyle{ecai}

\begin{abstract}
 Off-chain transaction channels represent one of the leading techniques to scale the transaction throughput in cryptocurrencies such as Bitcoin. They allow multiple agents to route payments through one another. So far, the topology and construction of payment networks has not been explored much. Participants are expected to minimize costs that are due to the allocation of liquidity as well as blockchain record fees. In this paper we study the optimization of maintenance costs of such networks. We present for the first time, a closed model for symmetric off-chain channels, and provide efficient algorithms for constructing minimal cost spanning-tree networks under this model. We prove that for any network demands, a simple hub topology provides a 2-approximation to the minimal maintenance cost showing that spanning trees in general are efficient.  We also show an unbounded price of anarchy in a greedy game between the transactors, when each player wishes to minimize his costs by changing the network's structure. Finally, we simulate and compare the costs of payment networks with scale free demand topologies. 
\end{abstract}
\section{Introduction}
A main approach to solve the scalability problem in Bitcoin is to use off-chain transaction
channels that allow parties to transfer funds while communicating directly, and only occasionally
to settle the accumulated transfers that they have done on the blockchain~\cite{LightningOriginal}.
In this manner, space on the blockchain (which is highly constrained) will be used to aggregate the result of multiple transactions that occurred over channels. 

Transaction channels are created between a pair of participants by locking money into a ``joint account''. The money is transfered between the participants by exchanging signed transaction messages that change the way funds in this account are allocated.  The channel is closed by the participants when sending the latest transaction to the blockchain, which finalizes this allocation. Channels can be chained together securely to allow users to transfer funds to others that they are not directly connected to, and thus form fast payment networks that provide near instant payments that require few interactions with the blockchain. 

The implication of this construction is that transaction channels necessitate locking liquidity for extended periods of time. These liquidity costs are balanced against the costs of fees that must be paid whenever channels are opened or closed---an operation that requires transactions on the blockchain. Intuitively, increasing liquidity costs by locking more funds into the channel lowers the fees that will be paid, as channels will need to be reset less often.

In this paper we take the first steps to study and optimize the costs of payment networks.
We study both game-theoretic formation of channels, and globally optimal solutions.

{\vskip 2mm} \noindent{ \bf Our Contributions:}
\begin{itemize}
  \item We show that hub topologies can be highly efficient in terms of costs minimization, which increases the concern that large monopolistic players will arise and out-perform decentralized networks. We show this by exploring cost-optimal topologies and liquidity allocations, and by proving that a hub topology provides a 2-approximation to the optimal maintenance cost of any topology. We further show that this approximation is tight.
  \item Using a model for the expected lifetime of symmetric channels (when both parties transfer coins in the same rate in both directions) from~\cite{DBLP:journals/corr/abs-1712-10222} we  derive the optimal liquidity allocation for any topology and routing algorithm.
  \item We solve the problem of building cost-optimal spanning tree networks for symmetric demands.
  \item Finally, we explore a game theoretic model for the network formation of payment channels by defining a game in which participants connect to the network and seek to minimize their costs. We study the equilibrium points in the restricted setting of trees for this game, and prove an unbounded price of anarchy for this setting.
\end{itemize}
Our results open the field for further explorations of the costs of payment networks, and highlight the need for further study of more general cases (such as non-symmetric demands).

\subsection{Payment Networks Basics}
Payment channels are based on exchanging signed messages that describe how to allocate funds between pair of agents.
For example, let us assume that Alice and Bob opened a payment channel by locking 0.5 BTC each.
When Alice wants to transfer 0.1 BTC to Bob, she signs a message containing the new channel balance (of 0.4 BTC for Alice, 0.6 BTC coins for Bob, as described in figure \ref{fig:aliceBobDemo}, transaction 1). She sends the message to Bob, that transmits the message to the blockchain at will. Bob will usually not do so immediately, but rather hold on to this message and replace it with newer allocations after more transfers are made.

While transaction channels themselves are limited to exchanges between pairs of individuals, further developments like the Lightning network\cite{LightningOriginal} allow users to securely route payments over longer paths and thus can allow the construction of a well connected network that can be used to transfer money quickly and with relatively little interaction with the blockchain. 
For example, let us assume that Alice wants to transfer money to Charlie. Instead of opening a shared payment channel or transacting directly via the blockchain, she can first transfer money to Bob via their shared payment channel, and ask Bob to send money to Charlie via another payment channel (as described in figure \ref{fig:aliceBobDemo}, transaction 2). In order to avoid the risk of having an intermediate node accept money without forwarding the payment, paths are composed using HTLCs (Hash time lock contracts) that release funds only if a secret---initially known only to the sender---is released. More details can be found in~\cite{LightningOriginal}. 

\begin{figure}[ht!]
  \centering
  \graphicspath{{lightning_demo.eps}}
  \includegraphics[scale=0.271]{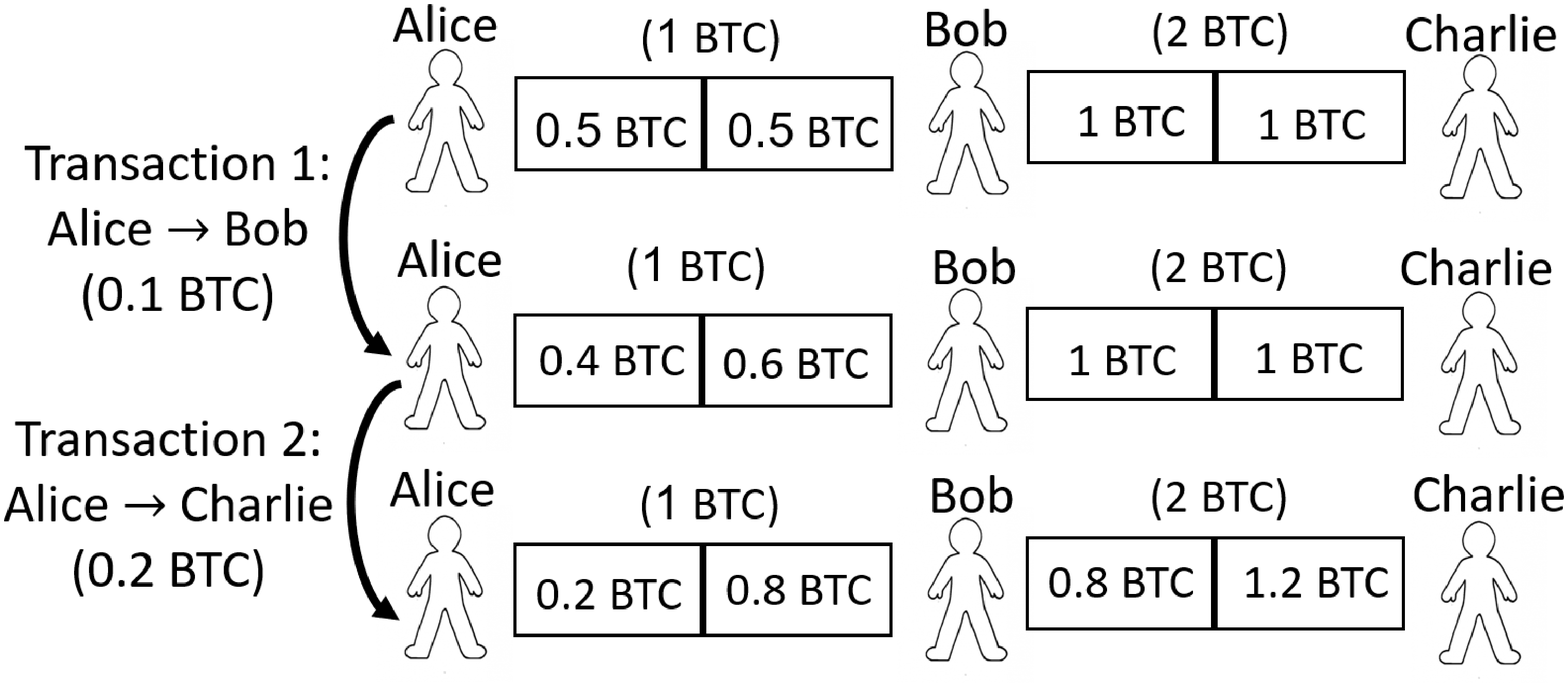}
  \caption{Draw of payment network usage}
  \label{fig:aliceBobDemo}
\end{figure}

When the most recent state of a channel allocates all funds in the joint account to one of the agents, it is no longer possible to transfer money to that agent using the channel. Further transfers must either use the blockchain directly, or the channel must go through a ``reset''; recording the channel balance
on the blockchain, and opening a new channel (alternatively funds can be added to the channel by the participant that lacks liquidity, but this too requires a blockchain transaction). As we have already mentioned, blockchain records incur a constant fee.

\subsection{Related work}
The original Lightning Network paper~\cite{LightningOriginal} suggests a BGP-like routing method. The implementation of the Lightning network (off-chain network over the Bitcoin blockchain), uses source routing, based on channel usage fees. The Raiden network~\cite{Raiden}, like other P2P networks, implements a Kademila-based routing system, in order to minimize the number of hops.
Another idea for scaling transactions over Ethereum, is the Perun payment hub \cite{dziembowski2019perun}.
The idea is based on a network in the structure of a hub, that will be efficient in the number of routing hops per transaction. The paper implements a secured method for shortcutting several hops, based on smart contracts.

Several recent papers ~\cite{avarikioti2019payment} ~\cite{avarikioti2019ride} ~\cite{ersoy2019profit} discuss the same question of payment channels creation game as we did. However, while those papers choose a transaction-series/single-channel based model, we differ with a network traffic rates model. A recent paper ~\cite{avarikioti2018payment} discussed the same question of global fees optimization in payment channels network design, but again it uses the same transaction-series model.

Two recent papers \cite{beres2019cryptoeconomic} \cite{rincon2020identifying} present a real data-driven analysis of lightning routing, fees and topology. We differ from them by presenting a theoretical model to this problem, rather than analyzing real-world data.

Some recent works \cite{routingHash1,routingHash2,herrera2019difficulty} refer to other side-effects of routing (such as privacy and memory demands), and present efficient algorithms for the routing problem.
The Fulgor and Rayo paper \cite{malavolta2017concurrency}, that describes two of the major methods for the remaining privacy and concurrency in off-chain networks, assumes that routing path is being chosen by the transactor, in order to minimize the fee per transaction. 

The topology of off-chain networks, has been studied in~\cite{adkinsinefficacy}. Measures such as the connectivity of the network and routing availability are the main focus. Algorithms for balancing edges (using ``smart routing'' of transactions) have been presented and simulated in~\cite{balancingAlgorithm}.
Some mathematical analysis of the expected channel lifetime over off-chain channels, has been conducted \cite{DBLP:journals/corr/abs-1712-10222}. We build upon and extend their model.

An economic model of the routing problem in Lightning, has been studied in~\cite{minimalPathArticle}, but we differ from this research in the basic model: The authors assume the cost of transmitting transactions is linearly related to the volume of payments and therefore the payoff can be formulated as a linear optimization problem. As far as we are aware, we are the first to explore cost minimization in this setting.

A recent paper \cite{khalil2017revive} presents an instance local solution to the channel resets problem, which is re-balancing a group of channels by transacting in cycles. We differ from this paper, by our approach of optimizing topology and liquidities, rather than re-balancing a liquidities in some specific channels.

The Plasma paper \cite{poon2017plasma} presents a generalization of payment channels to general state that is managed by local consensus protocols.

A few papers \cite{magnanti1984network} \cite{diaz2002survey} \cite{minoux1989networks} discuss the problem of optimizing communication networks from the perspective of cut-capacities, Gomory-Hu trees and submodular optimization. We use a similar approach in our paper.

\section{The Model}
We now present our model of off-chain payment networks.
We first define the state of a channel, and how transactions affect it, then define a payment network and a routing policy, and finally, a probabilistic model for transaction demand. 

\paragraph{The State of a Channel}
Let $e$ be a payment channel between two participants: Alice and Bob.
The payment channel state is a pair $(\omega_A, \omega_B)$, denoting the internal allocation of funds between the participants (Alice holds $\omega_A$ coins, and Bob holds $\omega_B$ coins).
The sum $\omega = \omega_A + \omega_B$ is called \emph{the liquidity of the channel}, and the amounts $\omega_A, \omega_B$ must be non-negative. 

\paragraph{Feasible Transactions} The liquidity of the channel, remains constant during its lifetime. Transactions, however, change the internal balance. 
Given a channel between Alice and Bob, with state $(\omega_A, \omega_B)$, a transaction of $x$ coins from Alice to Bob changes the state to $(\omega_A - x, \omega_B + x)$.
The transaction is \emph{feasible} if and only if $0 \le x < \omega_A$, i.e., it cannot reach or exceed the liquidity boundaries.

Although transfers in one direction are still possible when liquidity is fully shifted to one of the channel's sides, we assume the channel has to always be ready for transfers in both directions (otherwise we may be forced to use a blockchain transaction which takes very long). Therefore, we assume channels are reset whenever liquidity has fully shifted.

\paragraph{The Cost of Payment Channels}
We model the cost of locking liquidity $\omega$ in the channel as $\alpha \cdot \omega$ coins per second, where $\alpha$ reflects the interest rate in the economy. We further denote the blockchain fee required to write a record by $\phi$. This fee is paid whenever a channel reset occurs. \footnote{There may in fact be added costs for writing more-complex transactions to the blockchain that both transfer and reset the channel.}

We denote by $RPS$ the rate of blockchain records per second that are used to maintain channels. The cost of a channel (per time unit) is then $\alpha \cdot \omega + \phi \cdot RPS$.

We are now ready for definitions pertaining to more complex networks:
\begin{definition}[Payment Network]
  A payment network is described by a graph $G = (V,E)$, and a liquidity allocation $\omega: E \rightarrow R^{+}$.  
  $V$ denotes the set of transactors, $E$ is the set of payment channels, and $\omega(e)$ is the liquidity of channel $e$. We denote the sum of all channels' liquidities as $W$. Transactions are transfered over paths in the graph. A transaction of $x$ coins via the path $(e_1, ..., e_n)$, changes the state $(\omega_{i_1}, \omega_{i_2})$ of channel $e_i$ (for every $i$), to be $(\omega_{i_1} - x, \omega_{i_2} + x)$. Transactions must be feasible for all channels along the path.
\end{definition}

Given a payment network, it is important how payments are routed (as this may in turn affect the volume of payments going through individual channels).

\begin{definition}[Routing Policy]
  A routing policy $\pi$ over a graph $G= (V,E)$  is a function $\pi$ such that $\pi(i,j)$ is a legal path in $G$ between $i$ and $j$ with no cycles.
  For every $i,j \in V$ there holds $\pi(i,j)$ is the reverse of $\pi(j,i)$.
\end{definition}
This definition refers to a state-independent symmetric routing protocol. We restrict our attention to protocols that do not depend on direction or state, since this is the current way payment networks are implemented in Bitcoin (routing is independent of channel state as state changes too quickly to allow for global routing adjustments to occur).

Finally, we will define the demand matrix, which defines a probabilistic model for transaction requests:
\begin{definition}[The Demand Matrix]
  Let $G = (V,E)$ be a payment channel network, $V = {v_1, \ldots, v_n}$ and let $(\lambda_{i,j})$ denote the \emph{Demand Matrix} for transactions. Every second, $v_i$ transfers to $v_j$ an amount of coins, from a Poisson process with mean $\lambda_{i,j}$.   
\end{definition}
The demand matrix along with the topology and routing protocol effectively set the number of transfers on each channel and will later allow us to derive the maintenance costs of the network. 

\subsection{Symmetric Networks}
We are particularly interested in symmetric-demand networks, in which each pair of participants transact at the same expected rate in both directions, since such payments often cancel out and represent an optimistic case for transaction channels' lifetimes.

We thus assume that transaction amounts are all identical (for convenience, we assume all transfers are for 1 unit of money), and that furthermore, transaction rates are symmetric: $\lambda_{ij} = \lambda_{ji}, \ \ \forall \ 1\leq i,j\leq n$.   

The symmetric demands assumption, creates only balanced channels for every routing policy. That is because the traffic components over a channel, are identical in both directions. Under this model, each channel is simulating a random walk process. Every time a transaction through the channel is infeasible (due to liquidity constraints), we reset the channel, and pay a constant blockchain fee.

\begin{definition}
  Let $e$ be an off-chain transaction channel. Let $T_i$ [seconds] be the random variable of the channel lifetime after $i-1$ resets of the channel. Let $\phi$ [coins] be the constant fee for blockchain record. We define $RPS_{e}$ [$\frac{records}{seconds}$], the average rate of records per second, to be:
  $$
  RPS_e = \lim_{k\to\infty} \mathrm{E}\left[ \frac{k}{T_1 + ... + T_k}\right]
  $$
  The blockchain cost per second is then $\phi \cdot RPS_e$.
\end{definition}

\begin{lemma}
  Let $T_i$ [seconds] be the random variable of channel $e$'s lifetime after $i-1$ blockchain records. For every $i$, Let $\mu_e = E[T_i]$. Therefore: $$ \lim_{k\to\infty} \mathrm{E} \left[ \frac{k}{T_1 + ... + T_k}\right]= \frac{1}{\mu_e} $$
\end{lemma}

\begin{proof}
  From the Strong law of large numbers \cite{probBook}, since $T_1, \ldots, T_k$ are identically distributed, independent and have finite expected value, the series of arithmetic means $\frac{T_1 + ... + T_k}{k}$   is converging almost surely to $\mu_e$.
  We will prove that $ \frac{1}{\frac{T_1 + ... + T_k}{k}}$ converges almost surely to $\frac{1}{\mu_e}$. Let us denote: 
  \begin{equation} \begin{split} X_k = \frac{T_1 + ... + T_k}{k}, \ \  X = \mu_e, \ \ \\  A = \{\omega:  \lim_{k\to\infty} X_k(\omega) = X(\omega) \} \end{split}\end{equation}
  From convergence almost surely $P(A) = 1$. 
  From the basic laws of calculus, we know that on $A \cap \{ \omega: X(\omega) \neq 0\}$ there holds \begin{equation}\lim_{k\to\infty} \frac{1}{X_k(\omega)} = \frac{1}{X(\omega)}\end{equation}
  Since $X$ is a positive constant random variable, $\{ \omega: X(\omega) \neq 0\} = \Omega$. Therefore: \begin{equation} \begin{split} 1 \geq Pr\left(\lim_{k\to\infty} \frac{1}{X_k(\omega)}= \frac{1}{X(\omega)}\right) \\ \geq Pr(A \cap \{ \omega: X(\omega) \neq 0\}) = 1 \end{split} \end{equation}  From that we deduce $Pr(\lim_{k\to\infty} \frac{1}{X_k(\omega)} = \frac{1}{X(\omega)}) = 1$. Meaning that $ \frac{1}{\frac{T_1 + ... + T_k}{k}}$ is converging almost surely to $\frac{1}{\mu_e}$.
  We know that  $T_1, ... , T_k \geq 1$ because a channel reset requires at least one transaction. Therefore:
  $ \frac{T_1 + ... + T_k}{k} \geq 1$ which implies $ 0 < \frac{1}{\frac{T_1 + ... + T_k}{k}} \leq 1$ for all $k$.
  
  Let us denote a constant random variable $Y = 1$. We know $Y$ has finite expected value, and that $|\frac{1}{X_k}| \leq Y$ for all $k$.
  According to the Lebesgue dominated convergence theorem \cite{probBook}, Since  $\frac{1}{X_k}$ converges almost surely to $\frac{1}{\mu_e}$, we deduce:
  
  \begin{equation} \lim_{k\to\infty} \mathrm{E} \left[ \frac{k}{T_1 + ... + T_k} \right] = \mathrm{E} \left[\frac{1}{\mu_e}\right] = \frac{1}{\mu_e} \end{equation}
  
\end{proof}
The next corollary follows immediately:

\begin{corollary}
  $\displaystyle  \sum _{e \in E} {RPS_{e}} =  \sum _{e \in E} \frac{1}{\mu_e}$
\end{corollary}

From this formula we will be able to compute the blockchain records cost for every network and routing policy.

\paragraph{The Expected Lifetime Of Balanced Channels}
We recall a derivation from~\cite{DBLP:journals/corr/abs-1712-10222} that finds the optimal liquidity allocation and expected channel lifetime, in a balanced channel:

\begin{lemma}[From \cite{DBLP:journals/corr/abs-1712-10222}]
  Let us have a channel between Alice and Bob, with transfers of one coin occurring at a Poisson rate of $\lambda$ from Alice to Bob in each time slot, from Bob to Alice (a symmetric channel). The liquidity of the channel is $\omega$. The allocation that maximizes the expected time for the next reset is an equal allocation, and the expected value of channel life time is:
  $\mathrm{E} [T] =\frac{\omega^2}{8\lambda}$.
\end{lemma}

\paragraph{The Rate of Transfer Through a Channel}
In off-chain networks, every channel services multiple source-target transfers based on the routing policy. We define the rate over an edge (this rate is the same in both directions since our networks are symmetric): 
\begin{definition} \label{def:poisson_sum}
  Let $\pi$ be a routing policy over a network $G = (V,E)$, with a symmetric demands matrix $(\lambda_{ij})$. We denote for a channel $e \in E$, the balanced channel Poisson rate of $e$ as the following:  
  $\displaystyle \lambda_e = \sum_{\{i,j\} \ s.t \ e \in \pi(i,j)} {\lambda_{ij}} $
  
\end{definition}
This is essentially the total single-directional traffic over an edge, as derived from the routing policy. We rely here on the fact that the sum of Poisson random variables is a Poisson variable with rate equal to the sum of their rates. 

\section{Hub Is 2-Approximation for Cost Optimal Network}\label{appoxSection}

\begin{theorem}
  Let $G=(V,E)$ be an off-chain network, with a routing policy $\pi$, and liquidity allocation $\omega: E \rightarrow R^{+}$ . Let $H = (V, E_H)$ be a hub topology centered at an arbitrary vertex of $G$, $v_0 \in V$ .
  There exists a routing policy $\pi_H$ and liquidity allocation $\omega_H: E_H \rightarrow R^{+}$ for $H$, such that for every series of transactions, the number of channel resets is at most twice the number of channel resets in $G$, and the required liquidity is at most doubled: $\sum_{e'\in E_H} \omega_H(e') \leq 2 \cdot \sum_{e\in E} \omega(e)$.
\end{theorem}

  \begin{proof}
 The proof is based on duplicating original channel traffic in $G$ into two channels in $H$, by transferring the transaction throughout $v_0$ hub.
 For every node  $x\in V$, let us denote $E_x = \{ e\in E  \ s.t  \ x\in e\}$ all the channels in $G$ that $x$ is connected to. We define $\omega_H$ to be:
 \begin{equation}
 \omega_H(\{x, v_o\}) = \sum_{e\in E_x} \omega(e)
 \end{equation}
 We can now calculate the total amount of liquidity in $H$:
 \begin{equation}
 \begin{split}
 \sum_{e' \in E_H} \omega'(e') = \sum_{x\in V\setminus\{v_0\}}  \sum_{e\in E_x} \omega(e) \\ \leq \sum_{x\in V} \sum_{\{x,y\} \in E} \omega(\{ x,y\})  = 2 \sum_{e\in E} \omega(e)
 \end{split}
 \end{equation}
 Any channel's liquidity in the original graph, is added at most twice.
 Now we will define the hub routing policy $\pi_H$. Let $(v_1,v_2) \in V^2$ to be a pair of nodes. Let $P = (\{v_1, x_1\}, \{x_1, x_2\}, ..., \{x_k, v_2\}) = \pi(v_1, v_2)$  be the path in the original routing policy in $G$ (between $v_1$ and $v_2$). 
 The path $P_H = \pi_H (v_1, v_2)$ in the new routing policy over $H$ is: \begin{equation}P_H = (\{v_1, v_0\}, \{v_0, v_2\})\end{equation} 
 We remove all routing hops of the form $\{v_0, v_0\}$ from the new path $P_H$. 
 This is simply an elimination of the loop hops (trips back and forth to the hub).
 We will show that the number of channel resets is at most twice the original number of resets. We can look at each hub channel $\{ x,v_0\}$ as a ``superposition'' of all the original channels in $G$ that $x$ was connected to, $E_x$.
 This is also a superposition in the liquidity, since new liquidity is the sum of all $E_x$-channels liquidities.
 Therefore, we can split the original hub channel $\{x,v_0\}$, into imaginary channels that represent the original channels of $E_x$.
 We split the total liquidity $\sum_{e\in E_x} \omega(e)$ into the liquidities of imaginary channels of $E_x$. We simulate the action of each transaction over the imaginary channels in $E_x$, by changing their balances.
 
 Each channel that was originally between $x$ and $y$ now has a share of the liquidity in $\{x,v_0\}$ and in $\{v_0,y\}$.
 Every time that an imaginary channel has to be reset (in the simulation), we will perform a reset over that channel only in this both edges. It means that the hub's channel balance, summing the balances of all imaginary channels, will reset only the imaginary channel's balance component.
 This simulation fully implements a hub payment channel. In that way, for every series of transactions, the number of resets is the sum of all original $E_x$ channels reset numbers or less (due to the loops we've removed).
 
 When applying these partial resets, we can calculate the number of reset events over all the hub channels. It is bounded by the sum of imaginary $E_x$ channels reset events, for every $x$ in $V\setminus\{ v_0\}$.
 As a result, since every original $G$ channel can be summed at most twice (in two hub channels that connects to it's endpoints), the total number of resets in $H$ is at most twice the total number of resets in $G$. This result holds for every set of transactions, and therefore for all transaction demands.
 
\end{proof}

\begin{corollary}
  For any transaction demands, a hub over an arbitrary node is a 2-approximation for the optimal maintenance cost network.
\end{corollary} 

This corollary raises the concern that centralized topologies (such as a hub) will be used widely in off-chain networks due to their maintenance-cost efficiency, although they go against the decentralized nature of blockchain technologies.

\section{Optimizing Networks' Maintenance Costs}
In this section we will tackle the optimization problem seeking to minimize the rate of blockchain record fees per second (for a given amount of total liquidity). 

\subsection{Optimizing Liquidities Over a Network}
Here we derive the optimal allocation of the liquidity between the channels in any network, that minimizes the network's maintenance cost.

\begin{lemma}\label{lemmaLagruange}
  Let $G = (V,E)$ be a symmetric payment network, with a set of vertices $V$, and $n$ edges. Network channels have a $\lambda_e$ transaction rate in both directions for every edge $e$. The optimal liquidity allocation for $G$ (with total liquidity $W$) is:
  $ \omega(e) = W \frac{\lambda_e^\frac{1}{3}}{\sum_{\bar{e} \in E} \lambda_{\bar{e}}^\frac{1}{3} }$.
\end{lemma}

\begin{proof}
  Let $x_e = \frac{\omega_e} {W}$.
  The chain record fees we wish to minimize are: \begin{equation}\frac{8\phi}{W^2}\sum_{e\in E} \frac{\lambda_e}{x_e ^ 2}\end{equation} with a condition of $\sum_{e\in E} x_e = 1 $.\\
  We use Lagrange multipliers to optimize. Let us define $g$ by the condition $\sum_{e\in E} x_e = 1 $: 
  \begin{equation}g(x_1, ..., x_n, k) =  \frac{8\phi}{W^2}\sum_{e\in E} \frac{\lambda_e}{x_e ^ 2} - k\sum_{1\leq i\leq n} x_i\end{equation}
  If we differentiate $g$ and equate the derivatives to zero, we get \begin{equation}\frac{2\lambda_1}{x_1^3} = ... = \frac{2\lambda_n}{x_n^3}, \sum_{1\leq i\leq n} x_i - 1 = 0\end{equation}
  By that we get $ {x_i^3 \propto \lambda_i}$ (for every $i$). \\Since $x_i, \lambda_i > 0$ we conclude $ x_i\propto \lambda_i ^ \frac{1}{3} $ (for every $i$).
  
  We know $\sum_{1\leq i\leq n} x_i = 1$, which means we have to normalize the $x_i$'s. We have:
  
  \begin{equation}\omega(e) = Wx_e = W \frac{\lambda_e^\frac{1}{3}}{\sum_{\bar{e} \in E} \lambda_{\bar{e}}^\frac{1}{3} }\end{equation}
\end{proof}

\begin{corollary}\label{initialFormulaForCost}
  The optimal blockchain record fees per second for the network, with liquidity sum of $W$ is:
  $\frac{8\phi}{W^2} \left(\sum_{e \in E}{\lambda_e ^ \frac{1}{3}}\right)^3$.
\end{corollary}

\subsection{Cost Minimizing Spanning Trees}
Given a demand matrix for transactions, we wish to understand what the minimal cost payment network topology is. While solving this problem for general topologies is left open in this work, we are able to solve a more restricted case of tree topologies. 
The routing policy in a spanning tree is unique, because there is a unique path between every pair of vertices. This will allow us to more easily derive the best tree topology. 

First, we will calculate channels transaction rates, for a certain demand matrix and spanning tree. In order to do that, we reiterate the definition of cut-capacities from ~\cite{optComunicationTrees}:
\begin{definition}
  Let $G=(V,E)$ be a network, $(\lambda_{i,j})$ a capacity matrix, and $(X, V \setminus X)$ a graph cut.
  We define the cut's $\lambda$-capacity to be:
  $ \displaystyle
  \sum_{\{i,j\} \in E \ s.t \ i \in X, \ j \in V\setminus X} {\lambda_{i,j}}
  $
\end{definition}

Now we can calculate the transactions rate over a spanning tree channel:
\begin{lemma}\label{formTreeTraffic}
  The Poisson rate of a channel $e$, in a spanning tree network $T$ and demands matrix $(\lambda_{i,j})$, equals the $\lambda$-capacity of the cut between the two connectivity components of  $T-e$.
\end{lemma}

Once we can calculate the transaction rates in spanning trees, we adapt our setting to one that was defined originally by T.C Hu in 1974 \cite{optComunicationTrees}. In Hu's problem, we are given a demands matrix $(\lambda_{ij})$, denoting the communication requirements between pairs of nodes. Hu finds a spanning tree $T = (V,E_T)$, that minimizes $\sum_{e \in E_T} \lambda_e$ where $\lambda_e$ is the communication traffic over the tree's edge $e$. The communication rate over a tree channel, is similarly calculated according to lemma \ref{formTreeTraffic}.

The Gomory-Hu tree (cut-tree) \cite{gomoryHu}, can be calculated in polynomial time, and in addition to other properties, the tree is assured to minimize the sum of the tree's channels traffic, $\sum_{e\in E_T} \lambda_e$.

Unlike Hu's problem, we are interested in minimizing $\sum_{e\in E_T} (\lambda_e)^\frac{1}{3}$.
In this section, we will prove that the Gomory-Hu tree is also an optimal spanning tree for our problem. We will dive into the theory of submodular functions, in order to understand the properties of the Gomory-Hu tree.

\subsection{Submodular Optimization of Induced Record Fees}

It is known \cite{submodular} that the cut capacity function $c_\lambda(X,V\setminus X)$ for non-negative capacity function over the edges $\lambda: E \rightarrow R^+$, is a symmetric submodular function. 
The Gomory-Hu tree exists for every symmetric submodular function \cite{submodular}, when the definition is being adapted to a general function:
\begin{theorem}[Queyranne, Maurice \cite{submodular}]
  Let $V$ be a ground set, and let $f : 2^V \rightarrow\mathbf{R}^+$ be a symmetric submodular function. Given $s,t$ in $V$ define the minimum cut between $s$ and $t$ as 
  $\displaystyle \alpha_f (s,t) = \min_{W\subset V, |W\cap \{ s,t\}| = 1 } f(W) $.   
  Then, there is a Gomory-Hu tree that represents $\alpha_f$. That is, there is a tree $T = (V, E_t)$  and a liquidity function $c : E \rightarrow\mathbf{R}^+$ such that $\alpha_f (s,t) = \alpha_T (s,t)$ for all $s,t \in V$.
  Moreover, a minimum cut in T induces a minimum cut according to $f$ for every $s, t$. 
\end{theorem}
Since the cut capacity function $c_{\lambda}(X, V\setminus X)$ is submodular and symmetric, the Gomory-Hu tree that minimizes the sum of $n-1$ cut-capacities (created between the connectivity components of $T-e$ for every channel $e$), is assured to exist. However, we wish to optimize $\sum_{e \in E} \lambda_e ^ \frac{1}{3}$, and unlike the regular $\lambda$-cut capacity function, we cannot claim that $c_{\lambda}(X, V\setminus X)^\frac{1}{3}$  is submodular (although it is symmetric).
\footnote{In fact, a counter example to the submodularity of $c_{\lambda}(X, V\setminus X)^\frac{1}{3}$ exists and is provided in the full version.}
Quite surprisingly, we show that the Gomory-Hu tree $T$ for the \emph{regular cut function} minimizes $\sum_{e\in T} c_{\lambda}(C_e, V\setminus C_e)^\frac{1}{3}$.

\begin{lemma}\label{myAmazingLemma}
  Let $g: \mathbf{R} \rightarrow  \mathbf{R}$  be a non-decreasing function. Let $T$ be a Gomory Hu tree for a submodular symmetric non-negative function $f$ and a vertices set $V$.
  Then $T$ minimizes the value of $\sum_{e\in T} g(f(C_e)) $
  when $C_e$ is one connectivity component of $T-e$ (which creates a cut). 
\end{lemma} 
\begin{proof}
  We will use a property according to Adolphson and Hu \cite{adolphsonAndHu}.
  
  For two spanning trees $T_1 = (V, E_1), T_2=(V,E_2)$, there is a one-to-one mapping $\psi$ between $E_1$ and $E_2$, which satisfies the following condition:  
  For every $e = \{v_1, v_2\} \in E_1$,  $\psi(e)$ is in the path between $v_1, v_2$ in $T_2$.
  Let $T_1$ be a Gomory-Hu tree. For every edge $\{v_1, v_2\}=e_1\in E_1$ there is a unique edge $e_2\in E_2$ that represents a cut separating $v_1, v_2$ (since it is on the unique path between them in $T_2$).
  From the property of Gomory-Hu tree,  \begin{equation}\min_{W\subset V, |W\cap \{ v_1,v_2\}| = 1 } f(W) = f(C_{e_1})\end{equation}
  Thus, $f(C_{e_1}) \leq f(C_{e_2})$.
  Since $\psi$ is a one-to-one mapping \begin{equation}\sum_{e\in E_1} f(C_e)  \leq  \sum_{e\in E_1} f(C_{\psi(e)})  = \sum_{e\in E_2} f(C_e)\end{equation}
  Since $g$ is non decreasing,  we can apply $g$ over this equation:
  \begin{equation}\sum_{e\in E_1} g(f(C_e))  \leq  \sum_{e\in E_1} g(f(C_{\psi(e)}))  = \sum_{e\in E_2} g(f(C_e))\end{equation}
  Having that,  every spanning Tree $T$  will have greater or equal value of  $\sum_{e\in T} g(f(C_e))$, compared to the Gomory-Hu Tree.
\end{proof}

The main corollary (Gomory-Hu tree is optimal spanning tree network), is derived from the lemma with $f(X) = c_{\lambda}(X, V\setminus X), g(x) = x^{\frac{1}{3}}$:
\begin{corollary}
  Let $G = (V,E)$ be a complete graph with non-negative transaction Poisson rate over the edges $\lambda: E \rightarrow \mathbf{R}^+$.
  Let $T = (V, E_T)$ be a regular Gomory-Hu tree over the weight function $\lambda$.
  Then $T$ is the optimal spanning tree over the value of blockchain record fees per second, and it can be calculated in polynomial time.
\end{corollary}

\section{Towards Global Optimization}

\subsection{The Cost of a Symmetric Network}\label{subShikur}
The first step in exploring the global optimization of the maintenance cost, is to explore the connection between the different ingredients of that cost.

We showed in corollary \ref{initialFormulaForCost}, that the $RPS$ of a payment network is linear in $W^{-2}$, for every topology including the optimal. Therefore, we can denote the network's optimal $RPS$, as follows:
$ RPS = \frac{RPS_0}{W^2} $,
Where $RPS_0$ is the optimal RPS rate with $W = 1$ sum of liquidity.

Now, we can calculate the liquidity $W$ that minimizes the network's maintenance cost.
\begin{lemma}\label{lemmaShlishConnection}
  Let $G = (V,E)$ be a payment network, with $W$ liquidity sum. Let us denote $\alpha$ as the liquidity interest rate per second, $\phi$ as a single blockchain record fee, and $RPS_0$ as the optimal blockchain records per second rate (when $W = 1$). When optimizing the value of $W$, the optimal maintenance cost of the network is proportional to:   
  $$ M.C_{opt} \propto RPS_0 ^{\frac{1}{3}} \cdot \phi^{\frac{1}{3}} \cdot \alpha^{\frac{2}{3}}  $$
\end{lemma}

\begin{proof}
  According to our definition of channel costs, when summing the liquidity and blockchain records over all the channels in $E$, the full maintenance cost is:
  \begin{equation}
  M.C(W) = \phi \cdot \frac{RPS_0}{W^2} + \alpha \cdot W
  \end{equation}
  We can derive this equation, and find the minimum point (via the sign of the second derivative):
  \begin{equation}\frac{d}{dW} M.C = \alpha -\frac{2\cdot \phi \cdot RPS_0 }{W^3}, \\ W_{opt} = \left(\frac{2 \cdot \phi \cdot RPS_0}{\alpha}\right)^{\frac{1}{3}}\end{equation}
  By placing $W_{opt}$ into $M.C(W)$, we get the following relation:
  \begin{equation}
  M.C(W_{opt}) \propto  RPS_0 ^{\frac{1}{3}} \cdot \phi^{\frac{1}{3}} \cdot \alpha^{\frac{2}{3}} 
  \end{equation}
\end{proof}

\subsection {Tightness of Hub 2-Approximation}
The results of the previous section allow us to prove the tightness of the hub 2-approximation, that we showed in section \ref{appoxSection}.
\begin{theorem}
  For every $\epsilon > 0$, there exists a set of transaction demands $(\lambda_{ij})$, in which the optimal maintenance cost per second, is at least $(2-\epsilon)$ times lower than in every hub.
\end{theorem}

The full proof appears in the full version. The example that achieves a factor of 2-$\epsilon$ is based on a setting in which agents are paired together transacting only with each other. In a minimal cost network each such pair would form a channel between them, but in a hub, users are forced to route through two hops, incurring twice the cost. 

\section{Greedy Games for Payment Networks}
In this section we explore a "network formation" model that considers that players wish to minimize their individual costs. We begin by showing that the results of a greedy liquidity allocation, are about the same.

\subsection{Greedy Liquidity Allocation}
We analyze how individual channels optimize liquidity:
\begin{lemma}\label{coroFormula}
  The minimal maintenance cost per second for a balanced channel with $\lambda$ transactions Poisson rate, $\alpha$ interest rate for locked coins, and $\phi$ cost for a single Blockchain record, is:
  $M.C_{e} = 3\sqrt[3]{2} \lambda^\frac{1}{3} \phi^\frac{1}{3} \alpha^\frac{2}{3}$.
\end{lemma}

\subsection{Tree Topology Greedy Game}
In contrast to the liquidity allocation process, the greedy development of the network topology, is much more interesting. Each one of the players has an effect over the routing of other players. We will study a game, in which each player (node) wished to optimize his own costs, with restriction to the spanning tree topology. First, we define a player's share in channel fees:

\begin{definition}
  Let $G=(V,E)$ be a payment network. $e \in E$ is a channel with $M.C_e$ maintenance cost per second, and $\lambda_e$ balanced transaction Poisson rate. Let us assume that player $v \in V$ is transferring in $\lambda_e(v)$ rate over the channel $e$. Therefore, the share of player $v$ in the maintenance cost of $e$ is:
  $$ M.C_e(v) = \frac{\lambda_e(v)}{\lambda_e} \cdot M.C_e
  $$
\end{definition}

The player's fees are proportional to their relative part in the transaction rate. The motivation for this definition is that fees are charged as a fixed amount per transfer, that needs to cover the total cost of the channel.
We can now derive the total maintenance cost per second for a single player:

\begin{corollary}
  Let $\lambda_e(v)$ denote the transaction rate of player $v$ over channel $e$. The maintenance cost per second for player $v$ is then:
  $ \displaystyle M.C(v) = 3\sqrt[3]{2}\phi^\frac{1}{3} \alpha^\frac{2}{3} \sum_{e \in E} {\frac{\lambda_{e} (v)}{\lambda_{e}} \lambda_{e} ^ \frac{1}{3}}$
\end{corollary}

It is clear from the formula, that a player's costs depend on the routing policy. However, over a spanning tree topology, the routing is unique.

\begin{definition}
  The Tree Topology Greedy Game, is a game between the agents in $V$,
  assuming a certain demands matrix $(\lambda_{i,j})$.
  The goal of every player $v \in V$, is to minimize his maintenance costs per second, $M.C(v)$.
  
  The game starts with an initial spanning tree $T_0$. In every turn, each player can reconnect a single edge that is incident on its vertex. Such re-connections are restricted to maintain the spanning tree topology.
\end{definition}

\begin{definition}
  An equilibrium point in the spanning tree greedy game, is a spanning tree $T^*$, in which every player $v \in V$ cannot reduce his maintenance costs $M.C(v)$ by a legal move (of reconnecting a channel in the tree).
\end{definition}

\subsection{The Unbounded Price of Anarchy}\label{suboptimalSection}

\begin{lemma}
  The spanning tree topology game has unbounded price of anarchy.
  In other words, for every $K > 0$ there exists an equilibrium point in the spanning tree greedy game, $T^*$, in which there holds $M.C(T^*) > K \cdot M.C(T_G)$ where $T_G$ is the optimal maintenance cost spanning tree.
\end{lemma} 

\begin{proof}
 Let us build a series of equilibrium trees, $T_k$, with $k$ vertices, for every integer $k \geq 4$.
 The demands matrix is the following: $\lambda_{1, k} = 1$, $\lambda_{1,2} = 1$, and $\lambda_{k-1,k} = 1$. All the other transaction requirements equal zero.
 For every $k$, we can build a chain spanning tree - a cyclic ring along the nodes index, with removing the edge between $v_{k-2}$ and $v_{k-1}$. It creates 3 channels with $\lambda = 1$ transactions rate, and zero transactions rate for all the other $k-3$ channels.
 Now, let us describe the equilibrium spanning tree of $T_k$ by the set of edges $E_k$:
 \begin{equation}
 E_k = \{ \{v_1, v_2\} \{v_2, v_3\}, ..., \{v_{k-2}, v_{k-1}\}, \{v_{k-1}, v_k\} \}
 \end{equation}
 This is a simple chain from $v_1$ to $v_k$.
 By a quick calculation, we can see that using this spanning tree network, 2 channels will have $\lambda = 2$ transactions rate ($\{v_1, v_2\}$ and $\{v_{k-1}, v_{v_k}\}$).
 All the other $k-3$ channels will have a transaction rate of $\lambda = 1$.
Therefore, the ratio of maintenance cost per second, between $T_k$ and the optimal spanning tree, according to corollary \ref{coroFormula}, is at least:
\begin{equation}
\frac{(k-3) + 2^\frac{4}{3}}{3} 
\end{equation}
Since this expression is converging to infinity, for every value of $K$ there exists $k$, such that $M.C(T_k) > K \cdot M.C(T_G)$.
We will complete the proof by showing that $T_k$ is an equilibrium for every $k \geq 4$.

It is clear that the nodes from the series $v_3, ..., v_{k-2}$ has no interest to reconnect a channel connected with them, since their costs are fully payed by the remained nodes.
$v_1$ and $v_k$ are symmetric, as well as $v_2$ and $v_{k-1}$. It is clear these two pairs don't benefit from reconnecting, with both external nodes or with each other.

A precise examination shows that all the agents have no interest for structure change.
Therefore, $T_k$ is an equilibrium for every $k \geq 4$.
\end{proof}

\begin{figure}[ht!]
  
  \centering
  \graphicspath{{unbounded_equilibrium.eps}}
  \includegraphics[scale=0.3, trim={4cm, 7cm, 5cm, 0cm}]{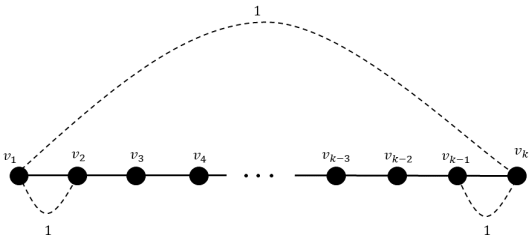}
  \caption{Drawing of $T_k$ tree (solid edges) and transactions requirements (dashed edges).}
  \label{fig:unboundedFig}
\end{figure}

\section{Simulation Results}\label{firstSimulations}
In this section we provide simulation results of more complex payment networks. Real-world transaction networks exhibit scale-free structure, and behave according to power-laws, e.g., in the distribution of degrees of vertices, and in transaction rates and amounts~\cite{beguvsic2018scaling}. One interesting question is whether in such cases, hubs perform better than a 2-approximation to optimal networks. 

We compare three network structures in our simulations: the optimal maintenance cost spanning tree (Gomory-Hu tree) derived with our algorithms, hub topologies (around an optimal node), and a complete graph topology. The latter two lie on two extremes: a centralized hub around a major transactor, or a channel between every pair of agents.

We simulated a network with 100 agents. We created the demand matrix in a scale-free graph model: generating a random scale-free graph~\cite{li2005towards} between the agents (with a certain power-law coefficient), and setting a non-zero transaction rate for every pair of connected nodes.
The non-zero transaction Poisson rates, were also generated from a power-law distribution for different power-law exponents.

We checked the ratio of expected blockchain records per second (RPS) between agents in each of the three scenarios (a complete graph, a hub, and an optimal spanning tree).

\begin{figure}[ht!]
  \centering
  \graphicspath{{deg_rps_ratio.eps}}
  \includegraphics[scale=0.34, trim={1.5cm, 0cm, 1cm, 1cm}]{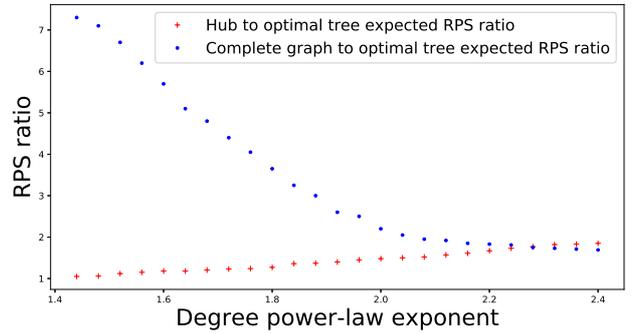}
  \caption{the ratio between hub and complete graph topologies to optimal spanning tree RPS v.s power law exponent of requirements graph degree distribution}
  \label{fig:RPS_deg}
  
\end{figure}

As we can see in figure \ref{fig:RPS_deg}, a complete graph topology is becoming more efficient as the transaction demands' centrality grows. This is because with a more centralized transaction demand graph, routing over a complete graph is closer to hub-routing. Moreover, it seems that in the majority of cases, a hub network achieves better results than the tight 8 RPS approximation (2-approximation in the total maintenance cost).
Further simulations show that transaction-rates' power-law coefficient has no affect over efficiency.

In the full version, we present additional simulation results in the same scale-free model, including simulations of the greedy game, and the price of anarchy. We show that power-law coefficients of the degree and transaction rates have no major affect on the probability of stability of optimal spanning trees, or on the price of anarchy. 

\section{Discussion and Future Work}
In this paper we defined a model for the maintenance cost of payment networks, and provided results on liquidity allocation when transaction demands are symmetric, and for optimal spanning trees. We further showed that spanning trees (and in particular hubs) provide a constant approximation for any optimal graph and routing system.
One of the weaknesses of our model is in the assumption of balanced channels. Some channels in payment networks might be unbalanced, and therefore their lifetime may be closer to a linear factor of the liquidity they hold, which will imply blockchain costs will be higher. 
Similarly, the question of different routing algorithms for asymmetric demands matrices, remains open and should be the focus of future work.


Another aspect that is still open, is the theory of online algorithms for balancing channels. Although some articles presented algorithms for this problem \cite{balancingAlgorithm}, they have not proposed a closed model for effectiveness estimation.
These balancing algorithms can reduce the channels lifetime significantly, and the analytical estimation question still remains open.


\bibliography{OptimizingOffChainNetworks}

\end{document}